\documentclass[11pt]{article}

\usepackage{algorithmic, algorithm}
\usepackage{amsmath,amsfonts}
\usepackage[margin=1in]{geometry}
\usepackage{amsthm} 
\usepackage[hypertex, colorlinks]{hyperref}

\newtheorem{theorem}{Theorem}[section]

\newtheorem{lemma}[theorem]{Lemma}

\newtheorem*{lemmaNoNum}{Lemma}
\newtheorem{claim}[theorem]{Claim}

\newtheorem{definition}[theorem]{Definition}
\newtheorem{remark}[theorem]{Remark}

\newcommand{\N}{{\mathcal{N}}}
\newcommand{\eqdef}{\stackrel{\Delta}{=}}

\newcommand{\R}{{\mathbb{R}}}
\renewcommand{\S}{{\mathbb{S}}}
\newcommand{\E}{{\mathbb{E}}}
\newcommand{\eps}{\varepsilon}
\newcommand{\xii}{x^{(i)}}

\newcommand{\seta}{\mathcal{A}}
\newcommand{\setb}{\mathcal{B}}

\newcommand{\ip}[1]{\left\langle #1 \right\rangle}
\newcommand{\SSS}{\mathbb{S}}
\newcommand{\abs}[1]{\left| #1 \right|}

\newcommand{\mnote}[1]{ \marginpar{\tiny\bf
            \begin{minipage}[t]{0.5in}
              \raggedright #1
           \end{minipage}}}

\usepackage{xspace}

\newcommand{\poly}{{\operatorname{poly}\xspace}}
\newcommand{\FHP}{$\operatorname{FHP}$\xspace}
\newcommand{\MMC}{$\operatorname{MMC}$\xspace}
\newcommand{\SVM}{$\operatorname{SVM}$\xspace}
\newcommand{\MAXTSAT}{$\operatorname{MAX-3SAT}$\xspace}
\newcommand{\MAXTSATt}[1]{$\operatorname{MAX-3SAT}(#1)$\xspace}
\newcommand{\SYM}{$\operatorname{SYM}$\xspace}
\newcommand{\SYMt}[1]{$\operatorname{SYM}(#1)$\xspace}

\date{\nonumber}
\begin{document}

\begin{titlepage}
\title{On the Furthest Hyperplane Problem\\ and Maximal Margin Clustering}

\author{
Zohar Karnin \thanks{Yahoo! Research {\tt zkarnin@yahoo-inc.com}.}
\and
Edo Liberty\thanks{Yahoo! Research, {\tt edo@yahoo-inc.com}.}
\and
Shachar Lovett \thanks{IAS, {\tt slovett@math.ias.edu}. Supported by DMS-0835373.}
\and
Roy Schwartz \thanks{Technion Institute of Technology and  Yahoo! Research  {\tt roys@yahoo-inc.com}.}
\and
Omri Weinstein\thanks{Princeton University and Yahoo! Research {\tt oweinste@cs.princeton.edu}.}
}

\maketitle


\begin{abstract}
This paper introduces the Furthest Hyperplane Problem (\FHP), which is an unsupervised counterpart of Support Vector Machines.
Given a set of $n$ points in $\R^d$, the objective is to produce the hyperplane (passing through
the origin) which maximizes the separation margin, that is, the minimal distance between the hyperplane and any input point.

%
To the best of our knowledge, this is the first paper achieving provable results regarding \FHP. We provide both lower and upper bounds to this NP-hard problem. 
First, we give a simple randomized algorithm whose running time is $n^{O(1/\theta^2)}$ where $\theta$ is the optimal separation margin. 
We show that its exponential dependency on $1/\theta^2$ is tight, up to sub-polynomial factors, assuming SAT cannot be solved in sub-exponential time. 
Next, we give an efficient approximation algorithm. 
For any $\alpha \in [0,1]$, the algorithm produces a hyperplane whose distance
from at least $1-5\alpha$ fraction of the points is at least $\alpha$ times the optimal separation margin. 
Finally, we show that \FHP does not admit a PTAS by presenting a gap preserving reduction from a particular version of the PCP theorem.

\end{abstract}

\end{titlepage}


\begin{section}{Introduction}
One of the most well known and studied objective functions in machine learning for obtaining
linear classifiers is the Support Vector Machines (\SVM) objective.
\SVM's are extremely well studied, both in theory and in practice.
We refer the reader to~\cite{VapLer63, Mangasarian65} and to~\cite{Burges98atutorial} for a thorough
survey and references therein. The simplest possible setup is the separable case:
given a set of $n$ points $\{\xii\}_{i=1}^n$ in $\R^d$ and labels $y_1,\ldots y_n \in \{1,-1\}$ find
hyperplane parameters $w \in \S^{d-1}$ and  $b \in \R$ which maximize $\theta'$ subject to
$(\langle w, x^{(i)} \rangle  + b)y_i \ge \theta'$. The intuition is that different concepts will be ``well
separated" from each other and that the best decision boundary is the one that maximizes the separation.
This intuition is supported by extensive research which is beyond the scope of this paper.
Algorithmically, the optimal solution for this problem can be obtained using Quadratic Programing or the
Ellipsoid Method in polynomial time. In cases where the problem has no feasible solution the constraints
must be made ``soft" and the optimization problem becomes significantly harder. This discussion,
however, also goes beyond the scope of this paper.


As a whole, \SVM's fall under the category of supervised learning although semi-supervised and
unsupervised versions have also been considered (see references below). We note that to the best of our knowledge the papers dealing with the unsupervised scenario were purely experimental and did not contain any rigorous proofs. In this model, the objective remains unchanged
but some (or possibly all) of the point labels are unknown. The maximization thus ranges
not only over the parameters $w$ and $b$ but also over the possible labels for the unlabeled points
$y_i \in \{1,-1\}$. The integer constraints on the values of $y_i$ make this problem significantly harder
than \SVM's.

In~\cite{Xu05maximummargin} Xu et al coin the name Maximal Margin Clustering (\MMC) for the case where
none of the labels are known. Indeed, in this setting the learning procedure behaves very much like clustering.
The objective is to assign the points to two groups (indicated by $y_i$) such that solving the
labeled \SVM problem according to this assignment produces the
maximal margin.\footnote{The assignment is required to label at least one point to each cluster to avoid a trivial unbounded margin.}
In~\cite{Bennett98semi-supervisedsupport} Bennett and Demiriz propose to solve the resulting
mixed integer quadratic program directly using general solvers and give some  
encouraging experimental results. De Bie et al~\cite{Bie03convexmethods} and Xu et al~\cite{Xu05maximummargin}
suggest an SDP relaxation approach and show that it works well in practice.
Joachims in~\cite{Joachims:1999} suggests a local search approach which iteratively improves on a current best solution.
While the above algorithms produce good results in practice, their analysis does not guaranty the optimality
of the solution. Moreover, the authors of these papers state their belief that the non convexity of this
problem makes it hard but to the best of our knowledge no proof of this was given.

\FHP is very similar to unsupervised \SVM or Maximum Margin Clustering.
The only difference is that the solution hyperplane is constrained to pass through the origin.
More formally, given $n$ points $\{\xii\}_{i=1}^n$ in a $d$-dimensional Euclidean space, \FHP is defined as follows:
\begin{eqnarray}\label{form1}
\textrm{Maximize}\quad    \theta' &   &    \nonumber \\
\textrm{s.t}\quad  \|w\|^2    & =& 1   \nonumber \\
	\forall \; 1\leq i \leq n \;\;\; |\langle w\cdot \xii \rangle |     & \geq & \theta'
\end{eqnarray}

The labels in this formulation are given by $y_i = sign(\langle w\cdot \xii \rangle)$ which can be viewed
as the ``side'' of the hyperplane to which $x^{(i)}$ belongs. At first glance, \MMC appears to be harder
than \FHP since it optimizes over a larger set of possible solutions. Namely, those for which $b$ (the
hyperplane offset) is not necessarily zero.
We claim however that any \MMC problem can be solved using at most ${n \choose 2}$ invocations of \FHP.
The observation is that any optimal solution for \MMC must have two equally distant points in opposite sides
of the hyperplane. Therefore, there always are at least two points $i$ and $j$ such that
$(\langle w , x^{(i)} \rangle  + b) = - (\langle w , x^{(j)} \rangle + b)$.
This means that the optimal hyperplane obtained by \MMC must pass through the point $(x^{(i)}+x^{(j)})/2$.
 Therefore, solving \FHP centered at $(x^{(i)}+x^{(j)})/2$ will
yield the same hyperplane as \MMC and iterating over all pairs of points concludes the observation.
From this point on we explore \FHP exclusively but the reader should keep in mind that any algorithmic claim
made for \FHP holds also for \MMC due to the above.


\begin{subsection}{Results and techniques}

In Section~\ref{edo_was_here} we begin by describing three exact (yet exponential) algorithms for \FHP.
These turn out to be preferable to one another for different problem parameters.
The first is a brute force search through all feasible labelings which runs in time $n^{O(d)}$.
The second looks for a solution by enumerating over an $\eps$-net of the $d$-dimensional unit sphere and requires $(1/\theta)^{O(d)}$ operations.
The last generates solutions created by random unit vectors and can be shown to find the right solution after $n^{O(1/\theta^2)}$ tries (w.h.p.).
While algorithmically the random hyperplane algorithm is the simplest, its analysis is the most complex.
Assuming a large constant margin, which is not unrealistic in machine learning applications, it provides the first polynomial time solution to \FHP.
Unfortunately, due to the hardness result below, its exponential dependency on $\theta$ cannot be improved.


In section~\ref{sec_approx} we show that if one is allowed to discard a small fraction of the points then much better results can be obtained. We note that in the perspective of machine learning, a hyperplane that separates almost all of the points still provides a meaningful result (see the discussion at the end of section~\ref{sec_approx}) \mnote{I added a small motivation here - see if you're ok with it since I have mixed feelings about it.}. We give an efficient algorithm which finds a hyperplane whose distance
from at least $1-5\alpha$ fraction of the points is at least $\alpha \theta$ , 
where $\alpha \in [0,1]$ is any constant and $\theta$ is the optimal margin of the original problem.
The main idea is to first find a small set of solutions which perform well `on average'. 
These solutions are the singular vectors of row reweighed versions of a matrix containing the input points.
We then randomly combine those to a single solutions. 


In section \ref{sec_hard} we prove that \FHP is NP-hard to approximate to
within a small multiplicative constant factor, ruling out a PTAS.
We present a two-step gap preserving reduction from MAX-3SAT using a particular version of the PCP theorem~\cite{arora_pcp}.
It shows that the problem is hard even when the number of points is linear in the dimension and when all the points have approximately the same norm.
As a corollary of the hardness result we get that the running time of our exact solution algorithm is, in a sense, optimal.
There cannot be an algorithm solving \FHP in time $n^{O(1/\theta^{2-\eps})}$ for any constant $\eps>0$, unless SAT admits a sub-exponential time algorithm.


\end{subsection}

\end{section} 


\begin{section}{Preliminaries and notations}

The set $\{\xii\}_{i=1}^n$ of input points for \FHP is assumed to lie in a
Eucledean space $\R^d$, endowed with the standard inner product denoted by $\ip{\cdot,\cdot}$.
Unless stated otherwise, we denote by $\|\cdot\|$ the $\ell_2$ norm.
Throughout the paper we let $\theta$ denote the solution of the optimization problem defined in
Equation~(\ref{form1}). The parameter $\theta$ is also referred to as ``the margin'' of  $\{\xii\}_{i=1}^n$
or simple ``the margin'' when it is obvious to which set of points it refers to. Unless stated otherwise,
we consider only hyperplanes which pass through the origin. They are defined by their normal vector
$w$ and include all points $x$ for which $\langle w,  x \rangle = 0$. By a slight abuse of notation, we
usually refer to a hyperplane by its defining normal vector $w$. Due to the scaling invariance of this
problem we assume w.l.o.g. that $\|\xii\| \leq 1$. One convenient consequence of this assumption is
that $\theta \leq 1$. We denote by $\mathcal{N}(\mu,\sigma)$ the standard Normal distribution with
mean $\mu$ and standard deviation $\sigma$. Unless stated otherwise, $\log()$ functions are base 2.
%
%
\begin{definition}[Labeling, feasible labeling] \label{def:labeling}
We refer to any assignment of $y_1,\ldots,y_n \in \{1,-1\}$ as a {\it labeling}. We say that a labeling is
feasible if there exists $w \in \S^{d-1}$ such that $\forall i \;  y_i\ip{w,x^{(i)}} > 0$.
Complementary, for any hyperplane $w \in \S^{d-1}$ we define its labeling as $y_i = sign(\ip{w,x^{(i)}})$.
\end{definition}

\begin{definition}[Labeling margin]
The margin of a feasible labeling is the margin obtained by solving \SVM on $\{\xii\}_{i=1}^n$ using
the corresponding labels but constraining the hyperplane to pass through the origin.
This problem is polynomial time solvable by Quadratic Programing or by the Ellipsoid Method~\cite{QP}.
We say a feasible labeling is optimal if it obtains the maximal margin.
\end{definition}
%
%
\begin{definition}[Expander Graphs]\label{expanders} An undirected graph $G = (V, E)$ is called an
$(n, d, \tau)$-expander if $|V | = n$, the degree of each node is $d$, and its edge expansion $h(G) =
\min_{|S| < n/2} (|E(S,S^c)|)/|S|$ is at least $\tau$. By Cheeger's inequality \cite{Cheeger}, $h(G) \geq
(d-\lambda)/2$, where $\lambda$ is the second largest eigenvalue, in absolute value, of the adjacency
matrix of $G$. For every $d = p + 1 \ge 14$, where $p$ is a prime congruent to $1$ modulo $4$, there
are explicit constructions of $(n, d, \tau)$-expanders with $\tau > d/5$ for infinitely many $n$. This is due
to the fact that these graphs exhibit $\lambda \leq 2\sqrt{d-1}$ (see~\cite{Ram}), and hence by the above
$h(G) \geq (d - 2\sqrt{d-1})/2 > d/5$ (say) for $d\geq 14$. Expander graphs will play a central role in
the construction of our hardness result in section \ref{sec_hard}.
\end{definition}

\end{section}


\begin{section}{Exact algorithms}\label{edo_was_here}


\begin{subsection}{Enumeration of feasible labelings}\label{enum}
The most straightforward algorithm for this problem enumerates over all feasible labelings
of the points and outputs the one maximizing the margin.
Note that there are at most $n^{d+1}$ different feasible labelings to consider.
This is due to Sauer's Lemma~\cite{Sauer1972145} and the fact that the VC dimension
of hyperplanes in $\R^d$ is $d+1$.\footnote{Sauer's Lemma~\cite{Sauer1972145} states that the number of possible
feasible labelings of $n$ data points by a classifier with VC dimension $d_{VC}$ is bounded by $n^{d_{VC}}$.}
This enumeration can be achieved by a Breadth First Search (BFS) on the graph $G(Y,E)$ of feasible labelings.
Every node in the graph $G$ is a feasible labeling ($|Y| \le n^{d+1}$)
and two nodes are connected by an edge iff their corresponding labelings defer by at most one point label.
Thus, the maximal degree in the graph is $n$ and the number of edges in this graph is at most $|E| \le |Y|n \le n^{d+2}$.
Moreover, computing for each node its neighbors list can be done efficiently since we only need to
check the feasibility (linear separability) of at most $n$ labelings.
Performing BFS thus requires at most $O(|Y| \poly(n,d) + |E|\log(|E|)) = n^{d + O(1)}$.
The only non trivial observation is that the graph $G$ is connected. To see this, consider the path from
a labeling $y$ to a labeling $y'$. This path exists since it is achieved
by rotating a hyperplane corresponding to $y$ to one corresponding to $y'$.
By an infinitesimal perturbation on the point set (which does not effect any feasible labeling)
we get that this rotation encounters only one point at a time and constitutes a path in $G$.
To conclude, there is a simple enumeration procedure for all $n^{d+1}$ linearly separable labelings which
runs in time $n^{d+O(1)}$.
\end{subsection}


\begin{subsection}{An $\eps$-net algorithm}\label{eps_net_alg}
The second approach is to search through a large enough set of hyperplanes and measure the margins produced by the labelings they induce.
Note that it is enough to find one hyperplane which obtains the same labels the optimal margin one does.
This is because having the labels suffices for solving the labeled problem and obtaining the optimal hyperplane.
We observe that the correct labeling is obtained by any hyperplane $w$ whose distance from the optimal one is $\|w - w^*\| < \theta$.
To see this, let $y^{*}$ denote the correct optimal labeling
$y^{*}_i\langle w, x^{(i)} \rangle = \langle w^*, y^{*}_i x^{(i)} \rangle + \langle w - w^*, y^{*}_i x^{(i)} \rangle  \ge \theta - \|w - w^*\|\cdot \|x^{(i)}\| > 0$.
It is therefore enough to consider hyperplane normals $w$ which belong to an $\eps$-net on the
sphere $\S^{d-1}$ with $\eps < \theta$. Deterministic constructions of such nets exist with size
$(1/\theta)^{O(d)}$ \cite{eps_net}.
Enumerating all the points on the net produces an algorithm which runs in time $O((1/\theta)^{O(d)} \poly(n,d))$.\footnote{
This procedure assumes the margin $\theta$ is known. This assumption can be removed by a standard doubling
argument.}
\end{subsection}


\begin{subsection}{Random Hyperplane Algorithm}\label{rha}
Both algorithms above are exponential in the dimension, even when the margin $\theta$ is large.
A first attempt at taking advantage of the large margin uses dimension reduction.
An easy corollary of the well known Johnson-Lindenstrauss lemma
yields that randomly projecting of the data points into dimension $O(\log(n)/\theta^2)$ preserves the margin up to a constant.
Then, applying the $\eps$-net algorithm on the reduced space requires only $n^{O(\log(1/\theta)/\theta^2 )}$ operations.
Similar ideas were introduced in \cite{Vempala} and subsequently used in \cite{Klivans04,Harpeled06,Balcan}.
%
%
%
It turns out, however, that a simpler approach improves on this.
Namely, pick $n^{O(1/\theta^{2})}$ unit vectors $w$ uniformly at random
from the unit sphere. Output the labeling induced by one of those vectors which maximizes the margin.
To establish the correctness of this algorithm it suffices to show that a random hyperplane induces
the optimal labeling with a large enough probability.
\begin{lemma} \label{lem:random FHP}
Let $w^*$ and $y^*$ denote the optimal solution of margin of $\theta$ and the labeling it induces.
Let $y$ be the labeling induced by a random hyperplane $w$.
The probability that  $y = y*$ is at least $n^{-O(1/\theta^{2})}$.
\end{lemma}
The proof of the lemma is somewhat technical and deferred to Appendix~\ref{app:fhp exact}.
The assertion of the lemma may seem surprising at first. The measure of the spherical cap of vectors $w$ whose distance from $w^*$ is at most $\theta$ is only $\approx \theta^d$. Thus, the probability that a random $w$ falls in this spherical cap is very small. However, we show that
it suffices for $w$ to merely have a weak correlation with $w^*$ in order to guarantee that (with large enough probability) it induces the optimal labeling.


Given Lemma \ref{lem:random FHP}, the Random Hyperplane Algorithm is straightforward. 
It randomly samples $n^{O(1/\theta^2)}$, computes their induced labelings, and output the labeling (or hyperplane) which admits the largest margin.
If the margin $\theta$ is not known, we use a standard doubling argument to enumerate it.
The algorithm solves \FHP w.h.p. in time $n^{O(1/\theta^2)}$.

\paragraph{Tightness of Our Result}

A corollary of our hardness result (Theorem~\ref{thm2}) is that, unless SAT has sub-exponential time algorithms, there exist no algorithm
for \FHP whose running time is $n^{O(\theta^{1/(2-\zeta)})}$ for any $\zeta>0$. 
Thus, the exponential dependency of the Random Hyperplane Algorithm on $\theta$ is optimal. 
This is since the hard \FHP instance produced by the reduction in Theorem~\ref{thm2} from SAT has $n$ 
points in $\R^d$ with $d=O(n)$ where the optimal margin is $\theta=\Omega(1/\sqrt{d})$. 
Thus, if there exists an algorithm which solves \FHP in time $n^{O(\theta^{1/(2-\zeta)})}$, 
it can be used to solve SAT in time $2^{O(n^{1-\zeta/2}\log(n))} = 2^{o(n)}$. 


\end{subsection}


\end{section} 


\begin{section}{Approximation algorithm} \label{sec_approx}
In this section we present an algorithm which approximates the optimal margin if one is allowed to discard a small fraction
of the points. For any $\alpha>0$ it finds a hyperplane whose distance from $(1-O(\alpha))$-fraction of the points
is at least $\alpha$ times the optimal margin $\theta$ of the original problem. 

Consider first the problem of finding the hyperplane whose {\em average} margin is maximal. 
That is, $w \in \S^{d-1}$ which maximizes $E_{i} \ip{w, x^{(i)}}^2$. 
This problem is easy to solve. 
The optimal $w$ is the top right singular vector the matrix $A$ whose $i$'th row contains $x^{(i)}$. 
In particular, if we assume the problem has a separating hyperplane $w^*$ with margin $\theta$, 
then $E_{i} \ip{w, x^{(i)}}^2 \ge E_{i} \ip{w^{*}, x^{(i)}}^2 \ge \theta^2$.
However, there is no guarantee that the obtained $w$ will give a high value of $|\ip{w, x^{(i)}}|$ for all $i$.
It is possible, for example, that  $|\ip{w, x^{(i)}}| =1$ for $\theta^2 n$ points and $0$ for all the rest.
Our first goal is to produce a set of vectors $w^{(1)},\ldots,w^{(t)}$ which are good on average for {\it every} point.
Namely, $\forall \;i \; E_j \ip{w^{(j)},x^{(i)}}^2 = \Omega(\theta^2)$.
To achieve this, we adaptively re-weight the points according to their distance from previous hyperplanes, so that those which have small inner product will have a larger influence on the average in the next iteration. We then combine the hyperplanes using random Gaussian weights in order to obtain a single random hyperplane which is good for any individual point.

\begin{algorithm}
\caption{Approximate \FHP Algorithm}
\label{alg apprx}
\begin{algorithmic}
\STATE {\bf Input:}  Set of points $\left\{x^{(i)}\right\}_{i=1}^{n} \in \R^d$
\STATE {\bf Output:}  $w \in \S^{d-1}$
\STATE $\tau_1(i) \leftarrow 1$ for all $i\in [n]$
\STATE $j \leftarrow 1$
\WHILE{$\sum_{i=1}^n \tau_j(i) \geq 1/n$}
        \STATE $A_j \leftarrow $ $n \times d$ matrix whose $i$'th row is $\sqrt{\tau_j(i)} \cdot x^{(i)}$
	\STATE $w^{(j)} \leftarrow$ top right singular vector of $A_j$
	\STATE $\sigma_{j}(i) \leftarrow \left|\ip{x^{(i)},w^{(j)}}\right|$
	\STATE $\tau_{j+1}(i) \leftarrow \tau_j(i)(1-\sigma^{2}_{j}(i)/2)$
	\STATE $j \leftarrow j +1$	
\ENDWHILE
\STATE $w'  \leftarrow \sum_{j=1}^t g_j \cdot w^{(j)}$ for $g_j \sim \N(0,1)$
\STATE \bf{return:} $w \leftarrow w'/\|w'\|$
\end{algorithmic}
\end{algorithm}

\begin{claim}
\label{claim:t small}
The main loop in Algorithm \ref{alg apprx} terminates after at most $t\leq 4\log(n)/\theta^{2} $ iterations. 
\end{claim}

\begin{proof}
Fix some $j$.
Define $\tau_j \eqdef \sum_{i=1}^n \tau_j(i)$.
We know that for some unit vector $w^*$ (the optimal solution to the \FHP) it holds that $|\ip{x^{(i)},w^*}| \ge \theta$ for all $i$.
Also since $w^{(j)}$ maximizes the expression $\|A_j w\|^2$ we have:
$$\sum_i \sigma^{2}_{j}(i) \tau_j(i) = \|A_j w^{(j)}\|^2 \geq \|A_j w^*\|^2 = \sum_i \tau_j(i) \cdot \ip{x^{(i)},w^*}^2 \ge \tau_j \cdot \theta^{2}.$$
It follows that $ \tau_{j+1} = \tau_j -  \sum_i \sigma^{2}_j(i) \tau_j(i) /2 \leq \tau_j(1-\theta^2/2) $ and the claim follows by elementary calculations since $\tau_1=n$.
\end{proof}

\begin{claim} \label{claim:sigma i large}
Let $\sigma_i \eqdef \sqrt{\sum_{j=1}^t \sigma^{2}_j(i)}$. When Algorithm \ref{alg apprx} terminates, for each $i$ it holds $\sigma_i^2 \geq \log(n)$.
\end{claim}

\begin{proof}
Fix $i \in [n]$, we know that when the process ends, $\tau_t(i) \leq \tau_t<1/n$. As $\tau_1(i)=1$ we get that
$$ 1/n \ge \tau_t(i) = \tau_1(i) \cdot \prod_{j=1}^t (1-\sigma^{2}_{j}(i)/2) =  \prod_{j=1}^t (1-\sigma^{2}_{j}(i)/2) \geq \prod_{j=1}^t 2^{-\sigma_{j}^{2}(i)}.$$
The last inequality holds since $1-x/2 \geq 2^{-x}$ for $0 \le x \le 1$. By taking logarithms from both sides, we get that
$\sum_j \sigma^{2}_{j}(i) \geq \log(n)$ as claimed.
\end{proof}

\begin{lemma}
\label{lemma:algorithm}
Let $0<\alpha<1$.  Algorithm \ref{alg apprx} outputs a random $w \in \S^{d-1}$ such that with probability at least $1/10$
at most an $5\alpha$ fraction of the points are such that $\abs{\ip{x^{(i)},w}} \le \alpha \theta$.
\end{lemma}

\begin{proof}[Proof of Lemma \ref{lemma:algorithm}]
First, by Markov's inequality and the fact that $\E[\|w'\|^2] = t$ we have that $\|w'\| \le 2\sqrt{t}$ w.p.\ at least $3/4$.
We assume this to be the case from this point on. 
Now we bound the probability that the algorithm `fails' for point $i$. 
\begin{eqnarray*}
\Pr\left[ \abs{\ip{w,x^{(i)}}} \le \alpha \theta\right] & \le & \Pr\left[\abs{\ip{w',x^{(i)}}} \le 2\sqrt{t}\alpha \theta\right]\\
&\le& \Pr_{Z\sim \N(0,\sqrt{\log(n)})}\left[ \abs{Z } \le 2\sqrt{t}\alpha \theta\right] \\
&=& \Pr_{Z\sim \N(0,1)}\left[| Z | \le \frac{2\sqrt{t}\alpha \theta}{\sqrt{\log(n)}}\right] \le  \frac{2}{\sqrt{2\pi}}\frac{2\sqrt{t}\alpha \theta}{\sqrt{\log(n)}} = \frac{8 \alpha}{\sqrt{2\pi}}
\end{eqnarray*} 
Since the expected number of failed points is less than $\frac{8 \alpha n}{\sqrt{2\pi}}$ we have, using Markov's inequality again, 
that the probability that the number of failed points is more than $5\alpha n$ is at most $0.65$. 
We also might fail with probability at most $0.25$ in the case that $\|w'\| > 2\sqrt{t}$.
Using the union bound on the two failure probabilities completes the proof.
\end{proof}



\paragraph{Discussion}

We note that the problem of finding a hyperplane that separates all but a small fraction of the points is the non-supervised analog of the well studied \emph{soft margin \SVM} problem. The motivation behind the problem, from the perspective of machine learning, is that a hyperplane that separates most of the data points is still likely to correctly label future points. Hence, if a hyperplane that separates all of the points cannot be obtained, it suffices to find one that separates most (e.g.\ $1-\alpha$ fraction) of the data points. The more common setting in which this problem is presented is when a separating hyperplane does not necessarily exist. In our case, although a separating hyperplane is guaranteed to exist, it is (provably) computationally hard to obtain it.

\end{section}


\begin{section}{Hardness of approximation} \label{sec_hard}
The main result of this section is that \FHP does not admit a PTAS unless P=NP.
That is, obtaining a $(1- \eps$)-approximation for \FHP is NP-hard for some universal constant $\eps$.
The main idea is straightforward: Reduce from \MAXTSAT for which such a guarantee is well known,
mapping each clause to a vector. We show that producing a ``far"
hyperplane from this set of vectors encodes a good solution for the satisfiability problem.
However, \FHP is inherently a symmetric problem (negating a solution
does not change its quality) while \MAXTSAT does not share this property.
Thus, we carry out our reduction in two steps: in the first step we
reduce \MAXTSAT to a symmetric satisfaction problem. In the second step we reduce this
symmetric satisfaction problem to \FHP. It turns out that in order to show that such a symmetric
problem can be geometrically embedded as a \FHP instance, we need the extra condition that
each variable appears in at most a constant number of clauses, and that the number of variables and
clauses is comparable to each other. The reduction process is slightly more involved in order to
guarantee this. In the rest of this section we consider the following satisfaction problem.

\begin{definition}[\SYM formulas]
A \SYM formula is a CNF formula where each clause has either $2$ or $4$ literals. Moreover,
clauses appear in pairs, where the two clauses in each pair have negated literals. For example,
a pair with $4$ literals has the form
\setlength\abovedisplayskip{3pt}
\begin{equation*}
(x_1 \vee x_2 \vee \neg x_3 \vee x_4) \wedge (\neg x_1 \vee \neg x_2 \vee x_3 \vee \neg x_4).
\end{equation*}
\setlength\abovedisplayskip{3pt}
We denote by $\textrm{\SYM}(t)$ the class of $\textrm{\SYM}$ formulas in which each variable
occurs in at most $t$ clauses.
\end{definition}
\noindent We note that \SYM formulas are invariant to negations: if an assignment $x$ satisfies $m$ clauses
in a \SYM formula than its negation $\neg x$ will satisfy the same number of clauses. \\

The first step is to reduce \MAXTSAT to \SYM with the additional property that each variable appears
in a constant number of clauses.
We denote by \MAXTSATt{t} the class of \MAXTSAT formulas where each variable appears in at most $t$ clauses.
Theorem~\ref{pcp} is the starting point of our reduction. It asserts that \MAXTSATt{13} is hard to approximate.

\begin{theorem}[Arora \cite{arora_pcp}, Hardness of approximating \MAXTSATt{13}]{\label{pcp}}
Let $\varphi$ be a 3-CNF boolean formula on $n$ variables and $m$
clauses, where no variable appears in more than $13$ clauses.
Then there exists a constant $\gamma >0$ such that it is NP-
hard to distinguish between the following cases: \vspace{.002in}
\begin{enumerate}
\setlength\abovedisplayskip{0pt}
\item$\varphi$ is satisfiable.
\setlength\abovedisplayskip{0pt}
\item No assignment satisfies more than a $(1 -\gamma)$-fraction of the
clauses of $\varphi$.
\end{enumerate}
\end{theorem}

\begin{subsection}{Reduction from \MAXTSATt{13} to \SYMt{30}}

The main idea behind the reduction is to add a new global variable to each  \MAXTSATt{13} clause
which will determine whether the assignment should be negated or not, and then
 to add all negations of clauses. The resulting formula is clearly a \SYM formula.
However, such a global variable will appear in too many clauses. We thus ``break" it into
many local variables (one per clause), and impose equality constraints between them.
To achieve that the number of clauses remains linear in the number of variables, we only impose equality
constraints based on the edges of a constant degree expander graph. The strong connectivity property of
expanders ensures that a maximally satisfying assignment to such a formula would assign the  same value
to all these local variables, achieving the same effect of one global variable.
%
%

We now show how to reduce \MAXTSAT to \SYM, while maintaining the property that each variable
occurs in at most a constant number of clauses.

\begin{theorem} {\label{thm1}}
It is NP-hard to distinguish whether a \SYMt{30} formula can
be satisfied, or whether all assignments satisfy at most $1-\delta$ fraction of the clauses,
where $\delta=\gamma/16$ and $\gamma$ is the constant in
Theorem~\ref{pcp}.  
\end{theorem}

\begin{proof}
We describe a gap-preserving reduction from \MAXTSATt{13} to \SYMt{30}. Given an
instance of \MAXTSATt{13} $\varphi$ with
$n$ variables $y_1,\ldots,y_n$ and $m$
clauses, construct a \SYM formula $\psi$ as follows:
each clause $C_i \in \varphi$ is mapped to a pair of clauses $A_i = (C_i \vee \neg z_i)$
and $A'_i = (C_i' \vee z_i)$ where $C_i'$ is the same as $C_i$ with all literals
negated and $z_i$ is a new variable associated only with the $i$-th clause.
For example:
$$ (y_1 \vee \neg{y_2} \vee y_3) \longrightarrow (y_1 \vee \neg{y_2}
\vee y_3 \vee \neg{z_i}) \wedge (\neg{y_1} \vee y_2 \vee \neg{y_3} \vee z_i).
$$
We denote the resulting set of clauses by $\seta$. We also add a set of  ``equality constraints",
denoted $\mathcal{B}$, between the variables $z_i$ and $z_j$ as follows. Let $G$ be an $(m,d,\tau)$
explicit expander with $d=14$ and $\tau \geq d/5$ (the existence of such constructions is established in
definition~\ref{expanders}). For each edge $(i,j)$ of the expander $\setb$ includes two
clauses: $(z_i \vee \neg{z_j})$ and $(\neg{z_i} \vee z_j)$. Let $\psi$ denote the conjunction
of the clauses in $\seta$ and $\setb$.

We first note that the above reduction is polynomial time computable; that $\psi$ contains
$M = (d + 2)m=16m$ clauses; and that every variable of $\psi$ appears in at most
$t:=max\{26, 2d+2\} = 30$ clauses. Therefore, $\psi$ is indeed an instance of \SYMt{30}.
To prove the theorem we must show:
\begin{itemize}
\item Completeness:
If $\varphi$ is satisfiable then so is $\psi$.
\item Soundness:
If an assignment satisfies $1-\delta$
fraction of $\psi$'s clauses then there is an assignment that satisfies $1-\gamma$ of $\varphi$'s clauses.
\end{itemize}

The completeness is straight-forward: given an assignment $y_1,\ldots,y_n$ that satisfies $\varphi$, we
can simply set $z_1,\ldots,z_m$ to $true$ to satisfy $\psi$. For the soundness, suppose that there exists
an assignment which satisfies $1-\delta$ fraction of $\psi$'s clauses, and let $v = y_1,\ldots,y_n,z_1,\ldots,z_m$
be a maximally satisfying assignment.\footnote{An assignment which satisfies the maximum possible number
of clauses from $\psi$.} Clearly, $v$ satisfies at least $1-\delta$ fraction of $\psi$'s clauses.
We can assume that at least half of $z_1,\ldots,z_m$ are set to $true$ since otherwise we can negate the
solution while maintaining the number of satisfied clauses.

We first claim that, in fact, \it all \rm the $z_i$'s must be set to true in $v$.  Indeed, let
$S=\{i: z_i = false\}$ and denote $k := |S|$  (recall that $k\leq m/2)$. Suppose $k>0$ and let $G$ be the
expander graph used in the reduction. If we change the assignment of all the variables in $S$ to
$true$, we violate at most $k$ clauses from $\seta$ (as each variable $z_i$ appears in
exactly $2$ clauses, but one of them is always satisfied). On the other hand, by definition
of $G$, the edge boundary of the set $S$ in $G$ is at least $\tau k = kd/5$, and every such
edge corresponds to a previously violated clause from $\setb$. Therefore, flipping the
assignment of the variables in $S$ contributes at least $kd/5 - k = \frac{14}{5}k - k > k$ to the number of
satisfied clauses, contradicting the maximality of $v$.

Now, since all the $z_i's$ are set to true, a clause $C_i \in \varphi$ is satisfied iff the clause
$A_i \in \psi $ is satisfied. As the number of unsatisfied clauses among $A_1,\ldots,A_m$ is
at most $\delta M = \delta (d+2) m$ we get that the number of unsatisfied clauses in $\varphi$
is at most $\delta(d+2) m = \frac{\gamma}{16}\cdot 16 m = \gamma m$, as required.
\end{proof}
\end{subsection}


\begin{subsection}{Reduction from \SYM to \FHP}

We proceed by describing a gap preserving reduction from \SYMt{t} to \FHP.

\begin{theorem}\label{thm2}
Given $\{x^{(i)}\}_{i=1}^{n} \in \R^d$, it is NP-hard to distinguish whether the furthest hyperplane
has margin $\frac{1}{\sqrt{d}}$ from all points or at most a margin of $(1-\eps) \frac{1}{\sqrt{d}}$
for $\eps=\Omega(\delta)$, where $\delta$ is the constant in Theorem~\ref{thm1}.
\end{theorem}
\begin{remark}
For convenience and ease of notation we use vectors whose norm is more than $1$ but at most $\sqrt{12}$.
The reader should keep in mind that the entire construction should be shrunk by this factor
to facilitate $\|x^{(i)}\|_2 \le 1$. Note that the construction constitutes hardness even for the special
case where $n=O(d)$ and for all points $1/\sqrt{12} \le \|x^{(i)}\|_2 \le 1$.
\end{remark}

\begin{proof}  Let $\psi$ be a \SYMt{t} formula with $d$ variables $y_1,...,y_d$ and $m$ clauses
$C_1,\ldots,C_m$. We map each clause $C_i$ to a point $x^{(i)}$ in $\R^d$. Consider first clauses
with two variables of the form $(y_{j_1}\vee y_{j_2})$ with $j_1<j_2$. Let $s_{j_1},s_{j_2} \in \{-1,1\}$
denote whether the variables are negated in the clause, where $1$ means not negated and $-1$
means negated. Then define the point $x^{(i)}$ as follows: $x^{(i)}_{j_1}=s_{j_1}$;
$x^{(i)}_{j_2}=-s_{j_2}$; and $x^{(i)}_{j}=0$ for $j \notin \{j_1,j_2\}$. For example:
$$
(y_2 \vee y_3) \longrightarrow (0,1,-1,0,\ldots,0).
$$
For clauses with four variables $y_{j_1},\ldots,y_{j_4}$ with $j_1<\ldots<j_4$ let $s_{j_1},
\ldots,s_{j_4} \in \{-1,1\}$ denote whether each variable is negated. Define the point $x^{(i)}$
as follows: $x^{(i)}_{j_1}=3 s_{j_1}$; $x^{(i)}_{j_r}=-s_{j_r}$ for $r=2,3,4$; and $x^{(i)}_{j}=0$
for $j \notin \{j_1,\ldots,j_4\}$. For example:
$$
(\neg y_1 \vee y_3 \vee y_4 \vee \neg y_6) \longrightarrow (-3,0,-1,-1,0,1,0,\ldots,0).
$$
Finally, we also add the $d$ unit vectors $e_1,\ldots,e_d$ to the set of points (the importance of
these ``artificially" added points will become clear later). We thus have a set of $n=m+d$ points.
To constitute the correctness of the reduction we must argue the following:
\begin{itemize}
\item Completeness:
If $\psi$ is satisfiable there exists a unit vector $w$ whose margin is at least $1/\sqrt{d}$.
\item Soundness:
If there exists a unit vector $w$ whose margin is at least $(1-\eps)/\sqrt{d}$ then there
exists an assignment to variables which satisfies $1-\delta$
fraction of $\psi$'s clauses.
\end{itemize}

We first show completeness. let $y_1,\ldots,y_d$ be an assignment that satisfies $\psi$.
Define $w_i=1/\sqrt{d}$ if $y_i$ is set to $true$, and $w_i=-1/\sqrt{d}$ if $y_i$ is set to
$false$. This satisfies $\|w\|_2=1$. Since the coordinates of all points $x^{(1)},\ldots,x^{(n)}$
are integers, to show that the margin of $w$ is at least $1/\sqrt{d}$ it suffices to show that $\ip{w,x^{(i)}}
\ne 0$ for all points. This is definitely true for the unit vectors $e_1,\ldots,e_d$. Consider now
a point $x^{(i)}$ which corresponds to a clause $C_i$. We claim that if $\ip{w,x^{(i)}}=0$ then
$y$ cannot satisfy both $C_i$ and its negation $C'_i$, which also appears in $\psi$ since it
is a symmetric formula. If $C_i$ has two variables, say $C_i=(y_1 \vee y_2)$, then
$x^{(i)}=(1,-1,0,\ldots,0)$, and so if $\ip{w,x^{(i)}}=0$ we must have $w_1=w_2$ and hence
$y_1=y_2$. This does not satisfy either $C_i=y_1 \vee y_2$ or $C'_i = \neg y_1 \vee \neg y_2$.
If $C_i$ has four variables, say $C_i = y_1 \vee y_2 \vee y_3 \vee y_4$, then
$x^{(i)}=(3,-1,-1,-1,0,\ldots,0)$, and so if $\ip{w,x^{(i)}}=0$ then either $w=(1/\sqrt{d})(1,1,1,1,\ldots)$
or $w=(1/\sqrt{d})(-1,-1,-1,-1,\ldots)$. That is, $y_1=y_2=y_3=y_4$, which does not satisfy either
$C_i$ or $C'_i$. The same argument follows if some variables are negated.

We now turn to prove soundness. Assume there exists a unit vector $w \in \R^d$ such that
$|\ip{w,x^{(i)}}| \ge (1-\eps)\frac{1}{\sqrt{d}}$. Define an assignment $y_1,\ldots,y_d$ as follows:
if $w_i \ge 0$ set $y_i=true$, otherwise set $y_i=false$. If we had that all $|w_i| \approx 1/\sqrt{d}$
then this assignment would have satisfied all clauses of $\psi$. This does not have to be the case,
but we will show that it is so for most $w_i$. Call $w_i$ whose absolute value is close to $1/\sqrt{d}$
``good", and ones which deviate from $1/\sqrt{d}$ ``bad". We will show that each clause which contains
only good variables must be satisfied. Since each bad variable appears only in a constant number of
clauses, showing that there are not many bad variables would imply that most clauses of $\psi$ are satisfied.
\begin{claim}
Let $B=\{i: |w_i - 1/\sqrt{d}| \ge 0.1 /\sqrt{d}\}$ be the set of ``bad" variables. Then $|B| \le 10 \eps d$.
\end{claim}

\begin{proof}
For all $i$ we have $|w_i| \ge (1-\eps)/\sqrt{d}$ since the unit vectors $e_1,\ldots,e_d$ are included in the
point set. Thus if $i \in B$ then $|w_i| \ge 1.1/\sqrt{d}$. Since $w$ is a unit vector we have

$$
1 = \sum w_i^2 = \sum_{i \in B} w_i^2 + \sum_{i \notin B} w_i^2 \ge
|B| \frac{1.1^2}{d} + (d-|B|)\frac{(1-\eps)^2}{d},
$$
\setlength\abovedisplayskip{3pt}
which after rearranging gives
$
|B| \le d \frac{1-(1-\eps)^2}{1.1^2-(1-\eps)^2} \le 10 \eps d.
$
\end{proof}
%
%
\begin{claim}
Let $C_i$ be a clause which does not contain any variable from $B$. Then the
assignment $y_1,\ldots,y_d$ satisfies $C$.
\end{claim}

\begin{proof}
Assume by contradiction that $C_i$ is not satisfied. Let $x^{(i)}$ be the point which
corresponds to $C_i$. We show that $\ip{w,x^{(i)}} < (1-\eps)/\sqrt{d}$ which
contradicts our assumption on $w$.

Consider first the case that $C_i$ contains two variables, say $C_i = (y_1 \vee y_2)$, which
gives $x^{(i)}=(1,-1,0,\ldots,0)$. Since $C_i$ is not satisfied we have $y_1=y_2=false$,
hence $w_1,w_2 \in (-1/\sqrt{d} \pm \eta)$ where $\eta < 0.1/\sqrt{d}$ 
which implies that $|\ip{w,x^{(i)}}| \le 0.2/\sqrt{d} < (1-\eps)/\sqrt{d}$. Similarly, suppose $C_i$ contains
four variables, say $C_i = (y_1 \vee y_2 \vee y_3 \vee y_4)$, which gives
$x^{(i)}=(3,-1,-1,-1,0,\ldots,0)$. Since $C_i$ is not satisfied we have
$y_1=y_2=y_3=y_4=false$, hence $w_1,w_2,w_3,w_4 \in (-1/\sqrt{d} \pm \eta)$ where $\eta < 0.1/\sqrt{d}$
which implies that $|\ip{w,x^{(i)}}| \le 0.6/\sqrt{d} < (1-\eps)/\sqrt{d}$. The other cases where some
variables are negated are proved in the same manner.
\end{proof}
\setlength\abovedisplayskip{7pt}

We now conclude the proof of Theorem~\ref{thm2}. We have $|B| \le 10 \eps d$. Since any
variable occurs in at most $t$ clauses, there are at most $10 \eps d t$ clauses which contain
a ``bad" variable. As all other clauses are satisfied,  the fraction of
clauses that the assignment to $y_1,\ldots,y_d$ does not satisfy is at most
$
10 \eps d t/m \le 10 \eps t < \delta
$
for $\eps=0.1 (\delta/t) = \Omega(\delta)$ since $t=30$ in Theorem~\ref{thm1}.
\end{proof}
\end{subsection}

\end{section}
\setlength\abovedisplayskip{0pt}


\begin{section}{Discussion}
A question which is not resolved in this paper is whether there exists an efficient constant
factor approximation algorithm for the margin of \FHP but for all points in the input. 
The authors have considered several techniques to try to rule out an $O(1)$ approximation
for the problem.
For example, trying to amplify the gap of the reduction in section \ref{sec_hard}.
This, however, did not succeed.
Even so, the resemblance of \FHP to some hard algebraic problems admitting no constant factor
approximation leads the authors to believe that the problem is indeed inapproximable to within a
constant factor.

\end{section}

\subsection*{Acknowledgments}
We thank Koby Crammer and Nir Ailon for very helpful discussions and references.

\bibliographystyle{unsrt}
\bibliography{furthestHyperplane}


\appendix

\pagebreak

\begin{section}{Proof of Lemma~\ref{lem:random FHP}} \label{app:fhp exact}

\begin{lemmaNoNum}
Let $w^*$ and $y^*$ denote the optimal solution of margin of $\theta$ and the labeling it induces.
Let $y$ be the labeling induced by a random hyperplane $w$.
The probability that  $y = y*$ is at least $n^{-O(\theta^{-2})}$.
\end{lemmaNoNum}

\begin{proof}
Let $c_1,c_2$ be some sufficiently large constants whose exact values will be
determined later. For technical reasons, assume w.l.o.g.\ that\footnote{If that is not the case to begin with, we can simply embed the vectors in a space of higher dimension.}
$d>c_1 \log(n) \theta^{-2}$.  Denote by $E$ the event that
$$ \ip{w^*,w}>\sqrt{c_2\log(n)}\theta^{-1} \cdot \sqrt{\frac{1}{d}} .$$
The following lemma gives an estimate for the probability of $E$. Although its proof is quite standard, we give it for completeness.

\begin{lemma} \label{lem:cap vol}
Let $w$ be a uniformly random unit vector in $\R^d$. There exists some universal constant $c_3$ such that for any $1\leq h\leq c_3\sqrt{d}$
and any fixed unit vector $w^*$ it holds that
$$ \Pr[\ip{w,w^*} > h/\sqrt{d}] = 2^{-\Theta(h^2)} .$$
\end{lemma}

As an immediate corollary we get that by setting appropriate values for $c_1,c_2,c_3$ we guarantee that $\Pr[E] \geq n^{-O(\theta^{-2})}$.

\begin{proof} [\bf Proof of Lemma \ref{lem:cap vol}\rm]
Notice that $\Pr[\ip{w,w^*} > h/\sqrt{d}]$ is exactly the ratio between
the surface area of a spherical cap defined by the direction $w^*$ and hight
(i.e., distance from the origin) $h/\sqrt{d}$ and the surface area of the
entire spherical cap. To estimate the probability we give a lower bound
for the mentioned ratio.

Define $S_d,C_{d,h}$ as the surface areas of the $d$ dimensional unit sphere
and $d$ dimensional spherical cap of hight $h/\sqrt{d}$ correspondingly. Denote by $S_{d-1,r}$ be the surface area of a $d-1$ dimensional sphere with radius $r$. Then,
$$ C_{d,h}/S_d =\int_{H=h/\sqrt{d}}^{1} \frac{S_{d-1,\sqrt{1-H^2}}}{S_d} dH  $$
We compute the ratio $\frac{S_{d-1,\sqrt{1-H^2}}}{S_d}$ with the well know formula for the
surface area of a sphere of radius $r$ and dimension $d$ of
$2\pi^{d/2}r^{d-1}/\Gamma(d/2)$ where $\Gamma$ is the so called Gamma
function, for which $\Gamma(d/2)=(\frac{d-2}{2})!$ when $d$ is even and
$\Gamma(d/2)=\frac{(d-2)(d-4)\cdots 1}{2^{(d-1)/2}}$ when $d$ is odd. We get
that for any $H<1/2$,
$$ \frac{S_{d-1,\sqrt{1-H^2}}}{S_d} = \Omega(\sqrt{d} \cdot (1-H^2)^{(d-2)/2} ) = \Omega(\sqrt{d} \cdot e^{-dH^2/2} ) $$
and that for any $H <1$,
$$ \frac{S_{d-1,\sqrt{1-H^2}}}{S_d} = O(\sqrt{d} \cdot (1-H^2)^{(d-2)/2} ) = O(\sqrt{d} \cdot e^{-dH^2/2} ) .$$
The lower bound is given in the following equation.
$$ \Pr\left[\ip{w,w^*} > h/\sqrt{d} \right] = C_{d,h}/S_d =  \int_{H=h/\sqrt{d}}^{1} \frac{S_{d-1,\sqrt{1-H^2}}}{S_d}  dH  \geq \int_{H=h/\sqrt{d}}^{2h/\sqrt{d}} \frac{S_{d-1,\sqrt{1-H^2}}}{S_d}  dH \stackrel{(*)}{=}$$
$$ \Omega \left( \int_{H=h/\sqrt{d}}^{2h/\sqrt{d}} \sqrt{d} \cdot e^{-dH^2/2}  dH \right) = \Omega\left( \int_{h'=h}^{2h}  e^{-h'^2/2}  dh' \right) = \Omega \left( h \cdot e^{-2h^2} \right) = e^{-O(h^2)}$$
Equation $(*)$ holds since $2h/\sqrt{d}<1/2$. The upper bound is due to the following.
$$ \Pr\left[\ip{w,w^*} > h/\sqrt{d} \right] = \int_{H=h/\sqrt{d}}^{1} \frac{S_{d-1,\sqrt{1-H^2}}}{S_d}  dH  = O \left( \int_{H=h/\sqrt{d}}^{1} \sqrt{d} \cdot e^{-dH^2/2}  dH \right) =$$
$$ O \left( \int_{h'=h}^{\infty} e^{-h'^2/2}  dh' \right) \stackrel{(**)}{=}  O \left( \int_{h'=h}^{\infty} e^{-h^2/2 - h h'}  dh' \right) = e^{-\Omega(h^2)} $$
In equation $(**)$ we used the fact that $h^2/2 + hh' \leq h'^2/2$ for all $h' \geq h$. The last equation holds since $h \geq 1$.
\end{proof}

We continue with the proof of Lemma~\ref{lem:random FHP}. We now analyze the success probability given the event $E$ has occurred.
For the analysis, we rotate the vector space so that $w^*=(1,0,0,\ldots,0)$. A vector $x$ can now be viewed as $x=(x_1,\tilde{x})$ where $x_1 = \ip{w^*,x}$ and $\tilde{x}$ is the $d-1$ dimensional vector corresponding to the projection of $x$ onto the hyperplane orthogonal to $w^*$. Since $w$ is
chosen as a random unit vector, we know that given the mentioned event $E$,
it can be viewed as $w=(w_1,\tilde{w})$ where $\tilde{w}$ is a uniformly
chosen vector from the $d-1$ dimensional sphere of radius $\sqrt{1-w_1^2}$
and $w_1 \geq \sqrt{c\log(n)}\theta^{-1} \cdot \sqrt{\frac{1}{d}}$.

Consider a vector $x\in \R^d$ where $\|x\|\leq 1$
 such that $\ip{w^*,x}\geq \theta$. As before we write
$x=(x_1,\tilde{x})$ where $\|\tilde{x}\| \leq \sqrt{1-x_1^2}$. Then
$$ \ip{x,w} = x_1 w_1 + \ip{\tilde{x},\tilde{w}} \geq \sqrt{\frac{c\log n}{d}} + \ip{\tilde{x},\tilde{w}} $$
Notice that both $\tilde{x},\tilde{w}$ are vectors whose norms are at most 1
and the direction of $\tilde{w}$ is chosen uniformly at random, and is independent of $E$. Hence, according to Lemma \ref{lem:cap vol},
$$ \Pr_w\left[\abs{\ip{\tilde{x},\tilde{w}}} \geq \sqrt{c\log n}/\sqrt{d}  \right]  \leq n^{-\Omega(c)} .$$

It follows that the sign of $\ip{w,x}$ is positive with probability
$1-n^{-\Omega(c)}$. By symmetry we get an analogous result for a vector $x$
s.t. $\ip{w^*,x}\leq -\theta$. By union bound we get that for sufficiently
large $c$, with probability $1/2$ we get that for all $i\in [n]$,
$sign{\ip{w,x^{(i)}}}=sign{\ip{w^*,x^{(i)}}}$ (given the event $E$ has
occurred) as required.  To conclude
$$ \Pr_{w\in \SSS^{d-1}}[ y=y^*] \geq \Pr_{w\in \SSS^{d-1}}[E] \cdot \Pr_{w\in \SSS^{d-1}}[ y=y^* | E ] \geq n^{-O(\theta^{-2})} .$$
\end{proof}

\end{section}





\begin{section}{An SDP relaxation for FHP}

The authors have considered the following SDP relaxation for the furthest hyperplane problem:

\begin{eqnarray}\label{sdp}
\textrm{Maximize}\quad    \|z\|^2 &   &    \nonumber \\
\textrm{s.t}\quad \forall \; 1\leq i \leq n \;\;\;\;  \|\sum_{j=1}^d \xii_j W^{j}\|^2    & \geq& \|z\|^2   \nonumber \\
		 \sum_{j=1}^d \|W^{j}\|^2  & = & 1
\end{eqnarray}

It is easy to see that the above semidefinite program is indeed a relaxation to FHP (given an optimal solution $w$
to FHP, simply set the first coordinate of $W^{j}$ to $w_j$ and the rest to zero, for all $j$. This is a feasible solution 
to the SDP which achieves value $\theta^2$). Nevertheless, this SDP has an integrality gap $\Omega(n)$.
To see this, observe that regardless of the input points, the SDP may always ``cheat" by choosing the vectors $W^{j}$ to
be orthogonal to one another such that $W^{j} = \frac{1}{\sqrt{d}}e_j$, yielding a solution of value $\|z^2\| = 1/d$. However, 
if $d=2$ and the $x^{(i)}$'s are $n$ equally spaced points on the unit circle of $\R^2$, then no integral solution has value 
better than $O(1/n)$. 
\footnote{In the same spirit, for general $d$, consider the instance whose $n$ points are all the points of an $\eps$-net on the unit 
sphere for $\eps = O(n^{-1/d})$. This instance has margin $O(\eps) = O(n^{-1/d})$, and therefore $\|z\|^2 = O(n^{-2/d})$, 
yielding an integrality gap of $\Omega(\frac{n^{2/d}}{\sqrt{d}})$.}
The question of whether this relaxation can be strengthened by adding convex linear constraints satisfied by the integral 
solution is beyond the reach of this paper.

\end{section}


\begin{section}{A note on average case complexity of \FHP}\label{random_theta'}

Given the hardness results above, a natural question is whether random instances
of \FHP are easier to solve. As our algorithmic results suggest, the answer to this
question highly depends on the maximal separation margin of such instances. We consider
a natural model in which the points $\{\xii\}_{i=1}^n$
are drawn isotropically and independently at random close to the unit sphere $S^{d-1}$.
More formally, each coordinate of each point is drawn independently at
random from a Normal distribution with standard deviation $1/\sqrt{d}$: $\xii_j \sim \mathcal{N}(0,1/\sqrt{d})$.

\noindent Let us denote by $\theta_{rand}$ the maximal separation margin of the set of points
$\{\xii\}_{i=1}^n$. While computing the exact value of $\theta_{rand}$ is beyond the reach of this paper
\footnote{The underlying probabilistic question to be answered is: what is the probability that $n$ random
points on $\S^{d-1}$ all fall into a cone of measure $\theta$ ?}
, we prove the following simple bounds on it:

\begin{theorem}
With probability at least $2/3$
$$\Omega\Big(\frac{1}{n\sqrt{d}}\Big) \;\; = \;\; \theta_{rand} \;\; = \;\; O\Big(\frac{1}{\sqrt{d}}\Big).$$
\end{theorem}

\begin{proof} For the upper bound, let $w$ be the normal vector of the furthest hyperplane achieving margin
$\theta_{rand}$, and let $y_i \in \{\pm 1\}$ be the sides of the hyperplane to which the points $\xii$
belong, i.e, for all $1\leq i \leq n$ we have $y_i\ip{w,\xii} \geq \theta_{rand}$.
Summing both sides over all $i$ and using linearity of inner products we get
\begin{eqnarray} {\label{rand2'}}
\ip{w,\sum_{i=1}^n y_i\cdot \xii} \geq \theta_{rand}\cdot n
\end{eqnarray}

\noindent By Cauchy-Schwartz and the fact that $\|w\| = 1$ we have that the LHS of \eqref{rand2'} is at most
$\|\sum_{i=1}^n y_i\cdot \xii\| = \|Xy\|$. Here $X$ denotes the $d\times n$ matrix whose $i$'th column is $\xii$,
and by $y$ the $\{\pm 1\}^n$ vector whose $i$'th entry is $y_i$.

 \begin{eqnarray} {\label{rand3'}}
\theta_{rand}\cdot n \leq \|Xy\| \leq \|y\|\cdot \|X\| \leq \sqrt{n}\cdot O\Big(\frac{\sqrt{n} + \sqrt{d}}{\sqrt{d}}\Big) =
O\Big(\frac{n}{\sqrt{d}}\Big)
\end{eqnarray}

\noindent where the second inequality follows again from Cauchy-Schwartz, and the third inequality
follows from the facts that the spectral norm of a $d \times n$ matrix whose entries are $\mathcal{N}(0,1)$
distributed is $O(\sqrt{n} + \sqrt{d})$ w.h.p. (see~\cite{LatalaRandomMatrices})
and the fact that $\|y\| = \sqrt{n}$. Rearranging \eqref{rand3'} yields the desired upper bound.

For the lower bound, consider a random hyperplane defined by the normal vector $w' / ||w'||$
where the entries of $w'$ distribute i.i.d. $\frac{1}{\sqrt{d}}\mathcal{N}(0,1)$. From the rotational invariance
of the Gaussian distribution we have that $\ip{w',\xii}$ also distributes $\frac{1}{\sqrt{d}}\mathcal{N}(0,1)$.
Using the fact that w.h.p $||w'|| \le 2$ we have for any $c>1$:
\begin{eqnarray} {\label{rand4'}}
 \Pr\Big[|\ip{w,\xii}| \leq \frac{1}{c\cdot n\sqrt{d}}\Big] \le \Pr\Big[|\ip{w',\xii}| \leq \frac{2}{c\cdot n\sqrt{d}}\Big] =
 \Pr_{Z \sim \mathcal{N}(0,1)}\Big[|Z| \leq \frac{2}{c\cdot n}\Big] = O\Big(\frac{1}{c\cdot n}\Big).
\end{eqnarray}
\noindent For a sufficiently large constant $c$, a simple union bound implies that the probability that
there exists a point $\xii$ which is closer than $1/(c\cdot n\sqrt{d})$ to the hyperplane defined by $w$
is at most $1/3$. Note that the analysis of the lower bound does not change even if the points are arbitrarily
spread on the unit sphere (since the normal distribution is spherically symmetric). Therefore, choosing a random
hyperplane also provides a trivial $O(n\sqrt{d})$ worst case approximation for \FHP. \end{proof}

\end{section}

\end{document}